\documentclass[12pt]{article}

\usepackage{amsmath,amsbsy,multirow,makeidx, siunitx, multirow,rotating}
\usepackage{amsfonts,amssymb,latexsym,bm,amsthm,lineno,mathtools}
\usepackage[margin=1in]{geometry}
\usepackage{booktabs,rotating,setspace}
\usepackage{natbib}
\usepackage{color}
\usepackage{dcolumn}
\usepackage{mathrsfs}
\usepackage{hyperref}
\usepackage{graphicx}
\usepackage{caption,subcaption}
\usepackage{lineno}

\newcommand{\sub}[1]{\ensuremath{_{\mbox{\scriptsize \,#1}}}}
\newcommand{\mN}{\mathcal{N}}

\newcommand{\E}{\ensuremath{\mathbb{E}}}
\newcommand{\R}{\ensuremath{\mathbb{R}}}
\newcommand{\V}{\ensuremath{\mathbb{V}}}
\newcommand{\Cov}{\ensuremath{\mathbb{C}\text{ov}}}

\newcommand{\mE}{\ensuremath{\mathcal{E}}}
\newcommand{\mX}{\ensuremath{\mathcal{X}}}
\newcommand{\mY}{\ensuremath{\mathcal{Y}}}

\newcommand{\convd}{\ensuremath{\xrightarrow{d}}}

\newcommand{\hta}{\ensuremath{\widehat{\theta}}}
\newcommand{\hf}{\ensuremath{\widehat{f}}}
\newcommand{\hh}{\ensuremath{\widehat{h}}}
\newcommand{\hnu}{\ensuremath{\hat{\nu}}}
\newcommand{\hmu}{\ensuremath{\hat{\mu}}}

\newcommand{\ba}{\ensuremath{\beta}}
\newcommand{\sig}{\ensuremath{\sigma}}
\newcommand{\ta}{\ensuremath{\tilde{a}}}
\newcommand{\tb}{\ensuremath{\tilde{b}}}
\newcommand{\tS}{\ensuremath{\widetilde{S}}}

\newcommand{\jphi}{J_{\emptyset}}
\newcommand{\hatm}{\widehat{m}}

\newtheorem{remark}{Remark}
\newtheorem{lemma}{Lemma}
\newtheorem{theom}{Theorem}

\DeclareMathOperator{\tr}{tr}
\DeclarePairedDelimiter\floor{\lfloor}{\rfloor}
\begin{document}

\title{The Effective Number of Parameters in Kernel Density Estimation}
\author{
Sofia Guglielmini\thanks{KU Leuven, Faculty of Economics and Business,
  Naamsestraat 69, 3000 Leuven, Belgium, sofia.guglielmini@kuleuven.be}
\and 
Igor Volobouev\thanks{Texas Tech University, Department of Department of Physics \& Astronomy, Lubbock TX 79409-1042, U.S.A., i.volobouev@ttu.edu}
\and 
A. Alexandre Trindade\thanks{Texas Tech University, Department of Mathematics
\& Statistics, Lubbock TX 79409-1042, U.S.A., alex.trindade@ttu.edu
(corresponding author)}
}
\maketitle

\begin{abstract}
We devise a formula for measuring the effective
degrees of freedom (EDoF) in kernel density estimation (KDE). Starting from the
orthogonal polynomial sequence (OPS) expansion for the ratio of the
empirical to the oracle density, we show how convolution with the kernel leads to a new OPS
with respect to which one may express the resulting KDE. The expansion
coefficients of the two OPS systems can then be related via a kernel
sensitivity matrix, which leads to a natural definition of EDoF through the trace operator. Asymptotic properties of the
plug-in EDoF are worked out through influence functions, and
connections with other empirical EDoFs are established. 
Minimization of Kullback–Leibler divergence is investigated as an
alternative to integrated squared error based
bandwidth selection rules, yielding a new normal scale rule. The
methodology suggests the
possibility of a new bandwidth selection rule  based on an information
criterion such as AIC. 
\end{abstract}


\section{Introduction}
In parametric modeling, the complexity of a model is measured by the number
of estimated parameters. In a nonparametric or semiparametric setting, this
concept is not as obvious, so the number of parameters (which is
undefined) could be replaced by the term \emph{effective degrees of
  freedom} (EDoF). An example is provided by kernel density estimation (KDE), also known as Parzen’s window
\citep{parzen1962kde}, which is arguably regarded as the default
nonparametric (smooth) density
estimation method \citep{silverman1986,ref:Wand1995,scott2015kde}.  The aim of this paper is to devise an expression for the EDoF in such kernel-based density estimation
procedures, which has historically been difficult to
quantify due to the infinite dimensionality of the underlying
function. The development of a plausible measure of model complexity
would then open up the possibility of new bandwidth selection rules based
on information criteria. 

To set the stage, let $F(x)$ and $f(x)$ be respectively the cumulative distribution and density functions of the continuous random variable
$X$, supported on the interval $[a,b]$. Although we allow $a$ and $b$ to be
infinite in the general problem formulation, some results will require
a compact support. Kernel-based methods essentially
smooth $f(x)$ according to a chosen kernel $K(x, y)$, resulting in
\begin{equation} \label{eq:oracle-convolution}
  f_K(y) = \int K(x, y) f(x) dx = \E K(X, y).
\end{equation}
Given an independent and identically distributed
(iid) sample $x_1, \ldots, x_n$ from $X$, KDE involves replacing the oracle density with
the empirical in \eqref{eq:oracle-convolution}:
\begin{equation} \label{eq:convolution}
  \hf(y) = \int K(x, y) \rho(x) dx = \frac{1}{n} \sum_{i=1}^n K(x_i, y),
\end{equation}
where $\rho(x)$ is the
empirical density expressed in terms of the Dirac delta function:
\begin{equation}\label{eq:empdensity}
\rho(x) = \frac{1}{n} \sum_{i=1}^n \delta(x - x_i).
\end{equation}

We take $K(x, y)$ to be a generalized kernel, not necessarily of the \emph{convolution} type
\begin{equation} \label{eq:convolution-kernel}
  K_h(x, y)=\frac{1}{h} K\left(\frac{y - x}{h}\right):=K_h(y-x),
\end{equation}
adding the subscript $h$ when it is desirable to emphasize its dependence on
the bandwidth tuning parameter $h$. We require the kernel to possess a few essential properties.
\begin{itemize}
\item[P1:] {\it Normalization}. $\int K(x, y) dy = 1$ for every $x$. This ensures that $\hf(y)$ is normalized.
\item[P2:] {\it Finite displacement}. For any function $g(x)$ supported on a finite interval $[a, b]$,
  the function $g_K(y) \equiv \int K(x, y) g(x) dx$,
is supported on some interval, $[\ta,\tb]$, which is also
finite. 
\end{itemize}
Although not essential, we list a few additional properties that may be useful/desirable. 
(P3) {\it Non-negativity}: $K(x, y) \ge 0$ for every $x$ and
$y$; ensures that the smoothed density is non-negative. (P4) {\it
  Double stochasticity}:  $\int K(x, y) dx = \int K(x, z) dx$ for
every $y$ and $z$; ensures that a uniform density remains uniform when
smoothed. (P5) {\it Bona fide smoother}: requires that for any probability density
  $g(x)$, the differential entropy of the density $g_K(y)$ in P2 is not
  decreasing in comparison with that of $g(x)$ \citep[ch.~8]{cover1991information}.

Early attempts at devising an EDoF formula include \cite{buja-etal-1989}, which
gives three possible definitions in the context of linear
smoothers. These ideas were echoed by  \cite{ye1998edof} and \citet[\S
7.6]{hastie-etal-2009esl} who proposed that the EDoF in a
general smoothing problem should be proportional to the sum of the covariances between
observed and fitted values. By analogy
with a regression model, the most
common variant reduces to the trace of the smoothing matrix.
This presents a problem for KDE since a design matrix is not
immediately apparent, however the recent paper by
\cite{ref:McCloud2020} purports to remedy that deficiency by
transforming the KDE to mimic a regression estimate.

Over twenty years earlier,  \cite{loader1996local} adapted the local
likelihood based smoothing idea of \cite{tibshirani1987local} to the
problem of density estimation. For a degree zero polynomial (local
constant fitting), Loader's extension coincides with KDE (under 
convolution kernels). In his monograph, \cite{ref:Loader1999} proposes
two definitions of EDoF for local likelihood density estimation, both based
on a linearization of the local likelihood estimation map. These
then lead directly to bandwidth selection rules based on information criteria such as
AIC.

Since the local constant version of local likelihood density
estimation coincides with KDE, the problem
of quantifying the EDoF is therefore essentially solved: we have
one version from \cite{ref:McCloud2020}, and two from
\cite{ref:Loader1999}. However, these versions are restrictive 
because they apply only under the convolution kernel setting, and
myopic due to the fact that they lack an oracle based formulation. Moreover,
the very fact that there are several versions implies that none of them
is definitive, thus opening up the possibility for further developments.   
In this paper we take a novel approach to the problem, by starting from the
orthogonal polynomial sequence (OPS) expansion for the ratio of the
empirical to the oracle density functions. Convolution with the kernel
then leads to a new OPS
with respect to which one may express the resulting KDE. The expansion
coefficients of the two OPS systems can then be related via a kernel
\emph{sensitivity} matrix, which essentially dictates how the
``noise'' in the empirical density is attenuated by the kernel
smoothing operation. This leads to a new plausible definition of
effective parameters by taking the trace of a symmetrized, positive
semi-definite, and normalized version.

This \emph{oracle} EDoF is
developed in Section~\ref{sec:oracle-edof}, with a corresponding empirical
version following by the plug-in principle in
Section~\ref{sec:emp-edof}. In Section~\ref{sec:connections} we
establish  connections among the versions of \cite{ref:Loader1999},
\cite{ref:McCloud2020}, and our empirical EDoF; while stressing that
ours is the only one that stems from an underlying oracle.
Section \ref{sec:amkld} establishes a connection between minimization
of Kullback–Leibler divergence and our EDoF.
Some examples where explicit calculations are possible are
provided in Sections \ref{sec-histogram} and \ref{sec-allgauss}, where
an explicit formula for EDoF is derived, resulting in a new normal
scale plug-in rule based on minimization of Kullback–Leibler divergence. 
Section \ref{sec-numerical} provides some numerical results on the performance of
the measures in controlled settings, and we conclude the paper with
an illustration on real data in Section \ref{sec-realdata}.

\section{Effective Degrees of Freedom: Oracle Version}\label{sec:oracle-edof}
As in the Introduction, let $f(x)$ be  the (oracle) density of $X$ supported on
the interval $[a,b]$. The  essence of our idea to devise an expression for EDoF, is to introduce an orthogonal
polynomial sequence (OPS) $\{P_j,\ j=0,1,\ldots\}$, that is orthonormal with
respect to a weight function comprised of $f(x)$ itself. This
leads to the representation:
\begin{equation}\label{eq:polydef}
  \int_a^b P_k(x) P_j(x) f(x) dx = \delta_{kj},
\end{equation}
where $\delta_{jk}$ denotes the Kronecker delta. Note that $P_0(x)=1$ and $P_j(x)$ is of degree $j$~\citep{ref:Gautschi}.
According to the Weierstrass approximation theorem, the OPS
defined by~\eqref{eq:polydef} represents a~complete basis for all continuous
square-integrable functions on $[a,b]$. For unbounded intervals, the OPS can usually be constructed if all
moments are finite, but the completeness property may or may not hold;
see~\cite{ref:Simon2010} for a~detailed exposition. To ensure completeness,
we initially restrict attention to compact intervals.

Expanding the ratio of the empirical density function to the oracle
density in terms of the OPS, we can formally write
  $\rho(x) = \sum_{i=1}^n \delta(x - x_i)/n =
              \sum_{k=0}^{\infty} c_k P_k(x)f(x)$, where 
\begin{equation}\label{eq:cj}
  c_j = \int_a^b \rho(x) P_j(x) dx = \frac{1}{n} \sum_{i=1}^n
  P_j(x_i),\quad \E (c_j) = \delta_{j0},\quad \Cov\left(c_k, c_j\right) = \frac{\delta_{kj} - \delta_{k0}\delta_{j0}}{n}.
\end{equation}
Now, note that from \eqref{eq:convolution} we can now write
$\hf(y) =\sum_{k=0}^{\infty} c_k \int_a^b K(x, y) P_k(x) f(x) dx$,
yielding an OPS expansion for the KDE in terms of
random zero-mean and uncorrelated coefficients (for $k\geq 1$ since
$c_0 = 1$). Introduce now an OPS with the weight
$f_K(y)$ defined by \eqref{eq:oracle-convolution}.
Assuming that $f_K(y)$ is supported on the interval $[{\tilde a}, {\tilde b}]$,
the polynomials $\{Q_j\}$ in this sequence are once again defined by the
orthonormality property
\begin{equation}\label{eq:polysfk}
  \int_{\tilde a}^{\tilde b} Q_k(y) Q_j(y) f_K(y) dy = \delta_{kj},
\end{equation}
with $Q_0(y) = 1$. In a similar fashion to what was done above, we can now expand \eqref{eq:convolution} in terms of this sequence:
\begin{equation}\label{eq:fhatrep}
  \hf(y) =  \left(\sum_{k=0}^{\infty} b_k Q_k(y) \right) f_K(y), \qquad  b_j = \int_{\tilde a}^{\tilde b} Q_j(y) \hf(y) dy = \sum_{k=0}^{\infty} c_k s_{jk},
\end{equation}
where the coefficients $\{b_k\}$ for $k\geq 1$ (since $b_0 = 1$) represent the {\it degrees of freedom} of $\hf(y)$, and  the $s_{jk}$ are the $(j,k)$ elements of the \textit{kernel sensitivity matrix} $S$, defined as:
\begin{equation}\label{eq:sensmat}
  s_{jk} = \int_a^b \int_{\tilde a}^{\tilde b} K(x, y) Q_j(y) P_k(x) f(x) dy dx.
\end{equation}
In terms of this matrix, the (infinite) vectors of coefficients $\{c_j\}$ and
$\{b_j\}$, representing the empirical density and the KDE respectively, are related
via $\vec{b} = S \vec{c}$. The term \emph{sensitivity} is used in the
sense of \cite{ye1998edof}, who proposes that a general definition of
EDoF for predicted or fitted values should measure the degree to which
these change in relation to small perturbations in observed values. 

Note that the covariance matrices of $\vec{b}$ and $\vec{c}$, $\Sigma_b$
and $\Sigma_c$, are related by $\Sigma_b = S \Sigma_c S^T$,
where the elements of $\Sigma_c$ are given by \eqref{eq:cj}. Integrating \eqref{eq:sensmat} over $y$ first with $Q_0(y) = 1$ and using the kernel
normalization property, we see that $s_{0k} = \delta_{0k}$. Similarly, $s_{j0} = \delta_{j0}$
(integrate over $x$ first, obtain $f_K(y)$,
and use \eqref{eq:polysfk}). This means that
both the first row and column of the kernel sensitivity
matrix are zero, except for the element $s_{00} = 1$. As $\V(c_0) = 0$,
this leads to 0 everywhere in the first row and column
of $\Sigma_b$. Thus, it suffices to focus attention on $\widetilde{S}$, the kernel sensitivity
matrix with the first row and column removed. Likewise, define
$\widetilde{\Sigma}_b$ and $\widetilde{\Sigma}_c$ to be corresponding
covariance matrices of $\vec{b}$ and $\vec{c}$ with the first row and column removed.

Borrowing the idea from  \cite{ye1998edof} and \citet[\S
7.6]{hastie-etal-2009esl} that the effective number of parameters in a
general smoothing problem is proportional to the sum of the covariances between
each observed value and the
corresponding  fitted value (heuristically embodied by the elements of
$\vec{c}$ and $\vec{b}$, respectively), leads to a possible definition of EDoF based on:
\[ \Cov(\vec{b},\vec{c}) = \Cov(S\vec{c},\vec{c}) = \frac{1}{n}S. \] 
A more precise definition would involve: (i) elimination of the
superfluous first row and column, (ii) multiplication by the transpose,
thus ensuring positive semi-definiteness, (iii) applying the trace operator,
thus summing all individual contributions, and (iv) normalizing, through scaling by the
covariance of $\vec{c}$. This leads to our proposed KDE-based estimate of the density $f$ with kernel $K$: 
\begin{equation}\label{eq:edof}
\text{EDoF} := \nu = n \tr(\widetilde{\Sigma}_b) = \tr(\widetilde{S}\widetilde{S}^T) = \sum_{j,k=1}^{\infty}s_{jk}^2= \sum_{j,k=0}^{\infty}s_{jk}^2-1.
\end{equation}
Viewing  $\widetilde{S}$ as a \emph{smoothing} matrix, the second trace formulation of this
expression also coincides with a definition for EDoF proposed earlier by
\cite{buja-etal-1989} in the context of linear smoothers. In words, the key idea is to consider the system as a collection
of uncorrelated ``noise sources'' (polynomial
terms in the empirical density expansion), and calculate EDoF by considering
how these noise sources are suppressed by smoothing.
While formally the matrix $\widetilde{S}$ is infinite-dimensional, in
practice elements with large row and column numbers appear to decay to zero quite quickly.
Therefore, for practical purposes it is sufficient to truncate $\widetilde{S}$
up to some reasonably low finite dimension. 

Instead of approximating $f$ with an OPS defined on $[\ta,\tb]$, an
interesting alternative to \eqref{eq:fhatrep} is obtained by mapping 
to the unit interval via the  $\{L_k\}$, Legendre polynomials
orthonormal on $[0, 1]$, so that
\begin{equation}\label{eq:fhatrepcd}
  \hf(y) =  \left(\sum_{k=0}^{\infty} b_k L_k(F_K(y)) \right) f_K(y),
\end{equation}
where $F_K(y)$ is the cumulative distribution function (cdf) of $f_K(y)$.
The expression in
parentheses of \eqref{eq:fhatrepcd} is known as the \textit{comparison
  density}~\citep[Ch.~2]{ref:Thas2010}. Performing the transformation
$z=F_K(y)$, and noting that $L_0(z) = 1$ and $b_0 = \int_{\tilde
  a}^{\tilde b} \hf(y) dy = 1$, the equation for $b_j$ in \eqref{eq:fhatrep} becomes
\begin{equation}\label{eq:bj-alt}
  b_j = \int_0^1 L_j(z) \frac{\hf(F^{-1}_K(z))}{f_K(F^{-1}_K(z))} dz =  \int_{\tilde a}^{\tilde b} L_j(F_K(y)) \hf(y) dy.
\end{equation}
Paralleling the development in~\eqref{eq:fhatrep} and~\eqref{eq:sensmat}, the elements of the
kernel sensitivity matrix can be alternatively defined by substituting $Q_j(y)=L_j(F_K(y))$:
\begin{equation}  \label{eq:altsensmat}
  s_{jk} = \int_a^b \int_{\tilde a}^{\tilde b} K(x, y) L_j(F_K(y)) P_k(x) f(x) dy dx.
\end{equation}
Using this expression,
we can relax the finite displacement kernel property P2,
as long as it can be assumed that the ``actual'' comparison density
in \eqref{eq:bj-alt}
is continuous and finite for all $z \in [0, 1]$,
so that the Weierstrass approximation theorem holds. Now note that
for this definition the results $s_{0k} = \delta_{0k}$ and
$s_{j0} = \delta_{j0}$  remain true, so that the corresponding
EDoF can still be defined by \eqref{eq:edof}.

Since \eqref{eq:altsensmat} involved making a change in the $y$ variable of
\eqref{eq:sensmat}, a similar switch to the comparison density can be made
for the $x$ variable instead via the
transformation $z=F(x)$ (the cdf of $X$).
Therefore, two additional definitions of the kernel sensitivity matrix
are given by
\begin{equation}\label{eq:sxlegendre}
   s_{jk} = \int_a^b \int_{\tilde a}^{\tilde b} K(x, y) Q_j(y) L_k(F(x)) f(x) dy dx,
\end{equation}
and
\begin{equation}\label{eq:salllegendre}
  s_{jk} = \int_a^b \int_{\tilde a}^{\tilde b} K(x, y) L_j(F_K(y)) L_k(F(x)) f(x) dy dx.
\end{equation}
While it still makes sense to define EDoF resulting from these four different
versions of $\widetilde{S}$ according to \eqref{eq:edof}, it is not
immediately obvious whether they will result in the same value
for $\nu$. The generalized Parseval Identity of Lemma
\ref{the:General-Parseval-Identity} settles this issue in the affirmative, by
noting that the four versions of $s_{jk}$ given by \eqref{eq:sensmat}, \eqref{eq:altsensmat},
\eqref{eq:sxlegendre}, and \eqref{eq:salllegendre}, are of
 the type in \eqref{eq:sjk-general} with $g(x,y)=K^2(x, y)/f_K(y)$.

\begin{lemma}[Bivariate Parseval Identity]\label{the:General-Parseval-Identity}
Let $g(x,y)$ be a continuous real-valued function
supported on the Cartesian product of the closed sets $\mX$ and $\mY$,
and let $\{P_k(x)Q_k(y)\}$, $k=0,1,2,\ldots$, be a complete bivariate OPS 
which is orthonormal with respect to~respective
weighting functions $dF_P(x)=f_P(x)dx$ and $dF_Q(y)=f_Q(y)dy$, with 
corresponding supports $\mX$ and $\mY$.   Consider the expression
\begin{equation}\label{eq:sjk-general}
  s_{jk} = \int_{\mY}\int_{\mX} g(x,y) P_j(x)Q_k(y) dF_P(x)
  dF_Q(y).
\end{equation}
Then, we have the bivariate Parseval Identity:
\begin{equation}\label{eq:app-General-Parseval-Identity}
\int_{\mY} \int_{\mX} g^2(x,y)dF_P(x) dF_Q(y) = \sum_{j,k=0}^{\infty}s_{jk}^2.
\end{equation}
\end{lemma}
\begin{proof}
See appendix. 
\end{proof}

The final issue to settle in this section is to obtain a more
closed-form expression for $\nu$, i.e., one that might be more
amenable to explicit computation.

\begin{theom}\label{the:edof-nu}
Let $f(x)$ be the density function of the continuous random variable
$X$ supported on the compact interval $[a,b]$, and $f_K(y)$ its
regularization in \eqref{eq:oracle-convolution}, supported on
$[\ta,\tb]$, by convolution with a kernel $K(x,y)$ satisfying properties P1 and P2. Let further $\{P_j(x)\}$ and $\{Q_j(y)\}$
be the OPS systems with respect to the weight functions $f(x)$ and
$f_K(y)$ respectively, and which satisfy the completeness condition of
Lemma \ref{the:General-Parseval-Identity}. Consider the expression for $\nu$ defined by \eqref{eq:edof}, where
the elements of the sensitivity matrix are given by
\eqref{eq:sensmat}. Then, we have that:
\begin{equation} \label{eq:app-directOracleEdof}
\nu = \int_a^b \int_{\tilde a}^{\tilde b} \frac{K^2(x, y)}{f_K(y)} f(x) dy dx - 1.
\end{equation}
\end{theom}
\begin{proof}
See the Appendix.
\end{proof}

\section{Effective Degrees of Freedom: Empirical Version}\label{sec:emp-edof}
Without knowledge of the oracle density $f(x)$, we can apply the same
procedure in which $f(x)$ is replaced everywhere by the empirical density
$\rho(x)$. This means that the OPS $\{P_j\}$ will be orthogonal with the
weight function $\rho(x)$, and consequently, the OPS $\{Q_j\}$ will be orthogonal with the
weight function $\hf(y)$, which is the empirical replacement for $f_K(y)$. Otherwise, the definition of the kernel sensitivity matrix
remains unchanged. The net effect of this is that the empirical EDoF,
$\hnu$, becomes simply the plug-in estimate for \eqref{eq:app-directOracleEdof}:
\begin{equation}\label{eq:directSampleEdof}
\hnu = \frac{1}{n} \sum_{i=1}^n \int_{\tilde a}^{\tilde b} \frac{K^2(x_i, y)}{\hf(y)} dy - 1.
\end{equation}

For the asymptotics and all results below, we now make it explicit
that the kernel is parametrized by bandwidth $h$. 
Denote the oracle and empirical versions of $\theta=1+\nu$ also with the
subscript $h$:
\begin{equation} \label{eq:OracleEdof-h}
  \theta_h = 1+\nu_h = \int_{\ta}^{\tb} \left\{\frac{\int_a^b
      K_h^2(x,y)f(x) dx}{\int_a^b K_h(x,y)f(x)dx}\right\}  dy =
  \int_{\ta}^{\tb}\frac{\E K_h^2(X,y)}{\E K_h(X,y)}dy = \int_{\ta}^{\tb}\frac{\mu_{2,h}(y)}{\mu_{1,h}(y)}dy,
\end{equation}
and
\begin{equation} \label{eq:EmpOracleEdof-h}
  \hta_h = 1+\hnu_h = \int_{\ta}^{\tb}\frac{\hmu_{2,h}(y)}{\hmu_{1,h}(y)}dy,
\end{equation}
where we use the notation: $\mu_{j,h}(y) =  \E K^j_h(X,y)$ and $\mu_{jk,h}(w, z) = \E K_h^j(X,w) K_h^k (X,z)$,
with corresponding empirical counterparts:
\[
\hmu_{j,h}(y) = \frac{1}{n}\sum_{i=1}^n K_h^j(x_i,y), \qquad\text{and}\qquad \hmu_{jk,h}(w, z) = \frac{1}{n}\sum_{i=1}^n K_h^j(x_i,w) K_h^k (x_i,z).
\]
\begin{remark}
Note that $\mu_{1,h}(y)\equiv f_K(y)$, and if $f(x)$ is supported on $[a,b]$, then $\mu_{j,h}(y)$ is
supported on $[\ta,\tb]$, with $\lim_{h\rightarrow 0}
[\ta,\tb]=[a,b]$. 
\end{remark}

To get a
finite sample approximation to the uncertainty of $\hta_h$, we can take
the limit as only $n\rightarrow\infty$, while $h$ is fixed. This is
easiest to derive using influence functions (IFs) and the
nonparametric delta method~\cite[Ch~2]{ref:Wasserman2006}. In order to
proceed, define the following linear functionals, and their
corresponding IFs:
\begin{alignat*}{2}
t_j(F;y) &= \mu_{j,h}(y) = \int K_h^j(x,y)dF(x), & \qquad  L_j(x,y) &=
K_h^j(x,y)-t_j(F;y), \quad j=1,2, \\ 
t_3(F) &= \int\log\mu_{1,h}(y)dF(y), & \qquad L_3(x) &=
           \log\mu_{1,h}(x) - t_3(F).  
\end{alignat*}
Note that both $t_1(F;y)$ and $t_2(F;y)$ depend on $y$.  For such
statistical functionals which contain also an argument $y$, $t(F;y)$, \cite{ref:KimScott2012} define the IF at $x$ as the following Gateaux derivative:
\begin{equation}\label{eq:influence-function}
L(x;y) =
\frac{\partial}{\partial\epsilon}\left.t\left((1-\epsilon)F+\epsilon H_x;y \right) \right|_{\epsilon=0},
\end{equation}
where $H_x$ denotes the Heaviside step function at
$x$. The empirical version, $\ell(x;y)$, is obtained by replacing
$F$ with the empirical cdf $\hf$. The following theorem derives the IF for
$t_4(F)=\theta_h$, and the accompanying limit theorem under the
$n\rightarrow\infty$ regime. 
\begin{theom}[IF and CLT for $\hta_h$]\label{the:clt-effdf} The IF for
  $t_4(F)=\theta_h$ is:
  \begin{equation*}
  L_4(x) =  \int_{\ta}^{\tb}L_4(x;y)dy, \qquad L_4(x;y) =
  \frac{K_h^2(x,y)}{\mu_{1,h}(y)}-\frac{\mu_{2,h}(y)}{\mu_{1,h}^2(y)}K_h(x,y). 
\end{equation*} 
Moreover, under the $n\rightarrow\infty$ regime (with $h$ fixed),
and provided $\omega^2_h<\infty$, we have the following asymptotic
normality result:
\begin{equation*}
 \sqrt{n}(\hta_h-\theta_h) \convd
 \mN\left(0,\omega^2_h\right),\qquad \omega^2_h = \int_{\ta}^{\tb}\int_{\ta}^{\tb}\frac{\gamma(w,z)}{\mu_{1,h}(w)\mu_{1,h}(z)}dwdz,
\end{equation*}
where
\[
\gamma(w,z) = \mu_{22,h}(w,z)-\frac{\mu_{2,h}(w)}{\mu_{1,h}(w)}\mu_{12,h}(w,z)-\frac{\mu_{2,h}(z)}{\mu_{1,h}(z)}\mu_{21,h}(w,z)+\frac{\mu_{2,h}(w)\mu_{2,h}(z)}{\mu_{1,h}(w)\mu_{1,h}(z)}\mu_{11,h}(w,z).
\]
\end{theom}
\begin{proof}
See the Appendix.
\end{proof}


\section{Connections With Other Versions of Empirical Effective
  Degrees of Freedom}\label{sec:connections}
As far as we can ascertain, there are only two references dealing with EDoF
explicitly for KDE, namely \citet[Ch.~5]{ref:Loader1999}, and
\cite{ref:McCloud2020}, neither of them postulating an oracle
version. The approach of \citet[Ch.~5]{ref:Loader1999}  is to approximate the
local likelihood cross-validation criterion involving the sum of
leave-one-out log-density estimates, by the sum of log-density
estimates minus an additional term, which naturally assumes the role of
a penalty factor in an information criterion based paradigm (e.g., AIC, BIC).
The idea of \cite{ref:McCloud2020} is to transform the usual KDE in
order to mimic a regression estimate, thus allowing for the extraction
of a hat (or projection) matrix whose trace can be used to define the
EDoF. There are close parallels between these developments and our
proposed empirical EDoF in \eqref{eq:directSampleEdof}, as we will now
demonstrate. 

\citet[ch.~5]{ref:Loader1999} defines two kinds of empirical EDoF, one
based on the influence function, the other on the variance function:
$\hnu_1 = \sum_{i = 1}^n \mbox{infl}(x_i)$, and $\hnu_2 = \sum_{i = 1}^n \mbox{vari}(x_i)$.
Since local likelihood density estimation of order zero, the so-called
``local constant'' approximation, coincides with boundary-corrected KDE, these versions are
directly applicable here. We now attempt to reformulate them for generalized kernels. The expression immediately below 
\citet[eqtn.~(5.14)]{ref:Loader1999}  states that
\begin{equation}\label{eq:infllogs}
  \sum_{i = 1}^n \log\,f_{-i}(x_i) \approx \sum_{i = 1}^n \log\,\hf(x_i) -  \sum_{i = 1}^n \mbox{infl}(x_i) + 1,
\end{equation}
where $\hf_{-i}(x_i)$ denotes the density estimate at $x_i$ with this
observation deleted. We can therefore determine $\mbox{infl}(x_i)$ by examining 
the difference $\log\hf(x_i)-\log\hf_{-i}(x_i)$.
Starting from $\hf(y)$ in \eqref{eq:convolution}, note that
\[
\hf_{-i}(y) = \frac{1}{n - 1} \sum_{\substack{j=1 \\ j\neq i}}^n K(x_j, y) = \frac{n}{n-1} \left[\hf(y) - \frac{1}{n} K(x_i, y)\right],
\]
whence $\hf(y) - \hf_{-i}(y) = (K(x_i, y) - \hf(y))/(n - 1)$.
Using a first order Taylor expansion of $\log\hf_{-i}(x_i)$ about
$\log\hf(x_i)$, we obtain ultimately that:
\[
  \sum_{i=1}^n \log\hf(x_i) - \sum_{i=1}^n \log\hf_{-i}(x_i) \approx \frac{1}{n-1} \sum_{i=1}^n \frac{K(x_i,x_i)}{\hf(x_i)} - \frac{n}{n-1}.
\]
According to \eqref{eq:infllogs} this should equal $\hnu_1 - 1$, and therefore
\begin{equation*}
  \hnu_1 = \frac{1}{n-1} \sum_{i=1}^n \frac{K(x_i,x_i)}{\hf(x_i)} + \left(1-\frac{n}{n-1}\right).
\end{equation*}
Dropping the term in parentheses (which is essentially zero and
otherwise inconsequential for the purposes of kernel tuning parameter
selection), and multiplying the result by the asymptotically
innocuous factor $(n - 1)/n$, yields
\begin{equation}\label{eq:nu1n}
  \hnu_1 \approx \frac{1}{n} \sum_{i=1}^n
  \frac{K(x_i,x_i)}{\hf(x_i)} = \hnu_3,
\end{equation}
where $\hnu_3$ is exactly the EDoF proposed by~\cite{ref:McCloud2020}.

According to \cite{ref:Loader1999} (and in particular the discussion
in \S~5.4.3),
$\mbox{vari}(x_i)$ used in his definition of $\hnu_2$ is actually the
variance of $\log\hf(x_i)$, i.e.,  $\V\log\hf(y)$ evaluated at
$y=x_i$. Assuming $\hf(x_i)>0$, the Delta
Method gives the following approximation:
\begin{equation} \label{eq:variX}
  \mbox{vari}(x_i) = \V\log\hf(x_i) \approx \frac{\V\hf(x_i)}{\left[\E\hf(x_i)\right]^2}.
\end{equation}
Now, seeking expressions for the first two moments of $\hf(x_i)$, note
that $\E\hf(y)=f_K(y)$, and
\begin{eqnarray*}
  \E\hf^2(y) &=& \frac{1}{n^2}\E\left[ \sum_{i,j=1}^n K(x_i, y)
                     K(x_j, y) \right] =\frac{1}{n^2}  \left[ n  \int K^2(x, y)f(x) dx + (n^2 - n) f_K^2(y)\right],
\end{eqnarray*}
which leads to the simplification $n\V\hf(y)=\int K^2(x, y)f(x) dx-f_K^2(y)$.
As $f(x)$ is unknown, further approximations must be made. In the
``local constant'' estimation paradigm, \cite{ref:Loader1999} makes the approximations:
$\E\hf(y) \approx \hf(y)$, and $\int K^2(x, y) f(x) dx \approx \hf(y) \int K^2(x, y) dx$,
whence \eqref{eq:variX} becomes
\[
  \mbox{vari}(x_i) \approx \frac{1}{n \hf(x_i)} \int K^2(x, x_i) dx - \frac{1}{n}.
\]
This results in an expression for $\hnu_2$ that is similar to our
proposed empirical EDoF:
\begin{equation}\label{eq:nu2n}
  \hnu_2 \approx \frac{1}{n} \sum_{i=1}^n
  \int\frac{K^2(x,x_i)dx}{\hf(x_i)}-1 \approx \frac{1}{n} \sum_{i=1}^n
  \int \frac{K^2(x_i, y)}{\hf(y)} dy - 1 = \hnu
\end{equation}
The formula given in \cite{ref:Loader1999} actually includes
additionally a boundary correction factor which does not appear here
(see e.g., Loader's Example 5.10). Instead,
we assume that $K(x, y)$ is already boundary-corrected (see Section~\ref{sec-numerical}).

The results in this section lead to plausible oracle versions for
$\hnu_1$ and $\hnu_3$. From \eqref{eq:nu1n}, we conclude
\begin{equation}\label{eq:nu3-oracle}
\nu_3 = \E\left(\frac{K(X,X)}{f_K(X)}\right) = \nu_1,
\end{equation}  
A corresponding result for $\nu_2$ does not however appear to be obvious.

\section{Proof of Concept: the Histogram}\label{sec-histogram}
In this section we show that the proposed EDoF measure leads to a
natural definition for the ``degrees of freedom'' when the ``model''
is a histogram. 
To obtain a histogram on the interval $[a, b]$ with $m$ bins
of variable widths $h_j$, define bin edges $b_0 = a$, $b_j = b_0 +
\sum_{k=1}^j h_k$, for $j=1,\ldots,m$, so that $b_m=b$ and bin $j$ consists of the
interval $B_j=[b_{j-1}, b_j)$. The histogramming kernel is then given
by: $K(x, y) = \sum_{j = 1}^m I(x \in B_j)I(y \in B_j)/h_j$,
where $I(x\in A)$ is the indicator function for set
$A$. Straightforward calculations then lead to the following results.
To compute our $\nu$ in \eqref{eq:app-directOracleEdof}, first notice
that
\[
  f_K(y) = \sum_{j=1}^m\frac{1}{h_j}I(y \in B_j)\int I(x \in B_j)f(x)dx =
  \sum_{j=1}^m \frac{p_j}{h_j}I(y\in B_j),
\]
where $p_j=P(X\in B_j)$. Calculation of the inner integral now gives:
\begin{align*}
	\int \frac{K^2(x,y)}{f_K(y)}dy &= \sum_{j=1}^m I(x\in B_j)
                                         \int_{B_j}\frac{I(y\in
                                         B_j)/h_j^2}{\sum_{k=1}^m
                                         p_kI(y\in B_k)/h_k}dy = \sum_{j=1}^m \frac{1}{p_j}I(x\in B_j).
\end{align*}
Finally, we obtain our oracle EDoF:
\begin{align*}
\nu &= \int \sum_{j=1}^m \frac{1}{p_j}I(x\in B_j) f(x)dx - 1 = \sum_{j=1}^m
   \frac{1}{p_j}\int I(x\in B_j)f(x)dx-1 = m-1.
\end{align*}

In an analogous way, we can obtain the empirical version, $\hnu$,
according to the definition in \eqref{eq:directSampleEdof}. Once
again, we start with the denominator in the definition,
\begin{align*}
 \hf(y) &= \frac{1}{n}\sum_{i'=1}^n K(x_{i'},y) = \frac{1}{n}\sum_{i'=1}^n\sum_{j=1}^m\frac{1}{h_j} I(x_{i'}\in B_j)I(y\in B_j)= \frac{1}{n}\sum_{j\in\jphi}\frac{n_j}{h_j}I(y\in B_j),
\end{align*}
where $n_j=\sum_{i=1}^nI(x_i\in B_j)$ is the number of observations
that fall in bin $j$, and $\jphi$ is the index set of all non-empty bins,
i.e., $\jphi=\{j=1,\ldots,m : n_j>0\}$. The ``inner'' integral analog
of the previous derivation now gives:
\begin{align*}
  \int \frac{K^2(x_i,y)}{\hf(y)}dy &= \sum_{j\in\jphi} \int
                                     \frac{I(y\in
                                     B_j)/h_j^2}{\frac{1}{n}\sum_{k\in\jphi}\frac{n_k}{h_k}I(y\in
                                     B_k)}dy\;I(x_i\in B_j) = \sum_{j\in\jphi} \frac{n}{n_j} I(x_i\in  B_j). 
\end{align*}
Thus, our empirical EDoF yields the value:
\begin{equation*}
\hnu+1 = \frac{1}{n}\sum_{i=1}^n\sum_{j\in\jphi} \frac{n}{n_j} I(x_i\in
B_j) = \sum_{j\in\jphi} \frac{n_j}{n_j} = \hatm,
\end{equation*}
where $\hatm$ is the cardinality of $\jphi$, i.e., the number of
non-empty bins.

Let us now compute the measure $\hnu_3$
of~\cite{ref:McCloud2020}. Starting with the above expression for
$\hf(y)$, note that evaluation at $x_i$ gives
\[
\hf(x_i) = \frac{1}{n}\sum_{j=1}^m\frac{I(x_i\in
  B_j)}{h_j}\sum_{i'=1}^n I(x_{i'}\in B_j) = \frac{1}{n}\sum_{j=1}^m\frac{n_j}{h_j}I(x_i\in
  B_j),
\]
so that
\[
\frac{1}{n}\frac{K(x_i,x_i)}{\hf(x_i)} = \frac{ \sum_{j=1}^mI(x_i\in
  B_j)/h_j}{\sum_{j=1}^mn_jI(x_i\in B_j)/h_j}.
\]
Summing this over all $i=1,\ldots,n$ leads to the result that $\hnu_3$ is exactly
the number of non-empty bins:
\[
\hnu_3 = \sum_{i=1}^n\frac{1}{n}\frac{K(x_i,x_i)}{\hf(x_i)} =
\sum_{j\in\jphi}\frac{n_j/h_j}{n_j/h_j} =\hatm.
\]

Near-identical calculations verify the expected result that the oracle
version of~\cite{ref:McCloud2020} yields $\nu_3=m$.
One can argue that while $\nu_3=m$ and $\nu=m-1$ are essentially
identical, the latter is the correct one, since it takes into account
the constraint that the area should sum to 1, leaving exactly $m-1$ ``degrees of freedom''.

\section{The Asymptotic Mean Kullback–Leibler Divergence (AMKLD)
  Criterion}\label{sec:amkld} 
In this section we introduce a new measure of optimality for KDE based
on minimizing the expected Kullback–Leibler divergence between $f(\cdot)$ and
$\hf(\cdot)$. Due to its tractability, the default measure in the KDE
literature has historically been the mean integrated squared error (MISE), with its asymptotic
approximation enjoying great success in the well-known result that
the AMISE-optimal bandwidth is $O(n^{-1/5})$. Although not as
tractable, our proposed AMKLD measure is interesting due to the fact
that the EDoF developed in \eqref{eq:edof} makes a surprising appearance. 

Starting from the Kullback–Leibler (KL) divergence between $f(y)$ and
$\hf(y)$, substituting for the latter by appealing to
\eqref{eq:fhatrep} leads to the representation: 
\begin{eqnarray}
\text{KL}(f||\hf) &=& -\int_a^b f(y) \log \left(\frac{\hf(y)}{f(y)}\right) dy \nonumber \\
                     & =& -\int_a^b f(y) \log \left(1 +
                          \frac{\ba(y)}{f(y)}\right) dy -\int_a^b f(y)
                          \log\left(\sum_{k=0}^{\infty} b_k
                          Q_k(y)\right) dy, \label{eq:full-kld}
\end{eqnarray}
where $\ba(y)=f_K(y)-f(y)$ is the \emph{bias} function. Now, treating
$\ba(y)/f(y)$ as small, Taylor series approximate the logarithmic term in the first
integrand of \eqref{eq:full-kld} using $\log(1+x)\approx x-x^2/2$, to give
\begin{equation}\label{eq:amkld-bias}
  -\int_a^b f(y) \log \left(1 + \frac{\ba(y)}{f(y)}\right) dy \approx
  -\int_a^b \ba(y) dy + \frac{1}{2} \int_a^b \frac{\ba^2(y)}{f(y)} dy
  = \frac{1}{2} \int_a^b \frac{\ba^2(y)}{f(y)} dy,
\end{equation}
where the equality follows from property P1 if we can assume $f(y)$ and $f_K(y)$ share
the same support, since in this case $\int_a^b f_K(y) dy =1= \int_a^b
f(y) dy$. Since $b_0 = 1$, $Q_0(z) = 1$, and the $b_k$ are
$O(n^{-1/2})$ for $k > 0$, we can similarly expand the second integral in
\eqref{eq:full-kld} to yield:
\begin{equation}\label{eq:amkld-variance}
  -\int_a^b f(y) \log\left(\sum_{k=0}^{\infty} b_k Q_k(y)\right) dy \approx -\sum_{k=1}^{\infty}  b_k  \int_a^b f(y) Q_k(y) dy + \frac{1}{2} \sum_{j,k=1}^{\infty} b_j b_k  \int_a^b f(y) Q_j(y) Q_k(y) dy.
\end{equation}
Now, from Section \ref{sec:oracle-edof}, note that, for $j,k>0$, we
have $\E(b_k) = 0$ and $n\Cov(b_j,b_k)=\widetilde{\Sigma}_b(j,k)$,
the $(j,k)$ elements of the covariance matrix of $\vec{b}$ (with the
first element of $\vec{b}$ excluded). With all the above assumptions,
and introducing the matrix $\mE$
whose  $(j,k)$ element is $\varepsilon_{jk} = \int_a^b \ba(y)  Q_j(y)
Q_k(y)dy$, we can express the second integral on the right-hand-side
of \eqref{eq:amkld-variance} as
$\int_a^b f(y)Q_j(y)Q_k(y)dy = \delta_{jk} - \varepsilon_{jk}$.

Therefore, to optimize
\emph{asymptotic mean KL divergence} (AMKLD), one needs to minimize
\begin{equation}\label{eq:amkld-def}
  \E\text{KL}(f ||\hf) \approx \mbox{AMKLD}(\hf) := \frac{1}{2}
  \int_a^b \frac{\ba^2(y)}{f(y)} dy + \frac{\nu}{2n} - \frac{\tr(\tS
    \tS^T\mE)}{2 n}\equiv B+V+BV,
\end{equation}
where we recall from \eqref{eq:edof} that $\nu=\tr(\tS \tS^T)$. In
this sum, the first term, $B$, represents the bias, the second term, $V$, the variance, and the
third term, $BV$, manifests what could be called \emph{bias-variance interaction}.
(In what follows, we refer to the sum of the second and the third terms
as the \emph{extended asymptotic variance}.) Simplifying expressions for the bias could now be invoked. For instance, in the case of a convolution kernel as in
\eqref{eq:convolution-kernel} with second moment $\mu_2(K)$, classical derivations lead to
$\ba(y)\approx\mu_2(K)f''(y)h^2$, with an error of $o(h^2)$ \citep{ref:Wand1995}.

\section{Exact Calculations: the All-Gaussian Case}\label{sec-allgauss}
Although the formula for $\nu$ in
\eqref{eq:edof} is not generally amenable to explicit calculation, we
have found one exception which sheds additional light on the meaning of EDoF. 
This is the case where the density of $X$ is $f(x)\equiv\phi(x;\sig)$, corresponding to a $N(0,\sig^2)$, and the
kernel is of the convolution type in \eqref{eq:convolution-kernel}
with $K(z)=\phi(z)$ the density of a standard normal. Since the convolution
of $\phi(x;\sig)$ with $K_h(y-x)=\phi(y-x;h)$, a Gaussian of variance $\sig^2$ with a
Gaussian of variance $h^2$, results in another Gaussian, the
density of the oracle convolution  corresponds to a normal with mean zero
and variance $\tau^2=\sigma^2 + h^2$, i.e., $f_K(y)=\phi(y;\tau)$.

Although direct calculation of $\nu$ according to
\eqref{eq:app-directOracleEdof} is more expedient here, we will proceed via the
original definition~\eqref{eq:edof} in order to reveal the structure of the kernel sensitivity matrix. For this case, the polynomials $P_k(x)$ orthonormal with weight $\phi(x;\sig)$
are
\[
  P_k(x) \equiv Hn_k(x, \sigma) = \frac{1}{2^{k/2} \sqrt{k!}}
  H_k\left(\frac{x}{\sqrt{2} \sigma}\right),
\]
where the $H_k(z)$ are the standard ``physicist's'' Hermite
polynomials orthogonal with weight function
$\exp\{-z^2\}$ \citep[\S
22.2.14]{abramowitz-stegun-1972}. The convolution of two Gaussians then
facilitates the computation of the derived OPS $\{Q_j\}$ orthonormal
with respect to the weight $\phi(y;\tau)$, so that $Q_j(y) = Hn_j(y, \tau)$.
The elements of the kernel sensitivity matrix defined by \eqref{eq:sensmat} therefore become:
\begin{equation}\label{eq:gausssens}
  s_{jk} = \frac{1}{2 \pi h \sigma} \int_{\R^2} Hn_j(y, \tau) \,Hn_k(x, \sigma) \, \exp\left\{-\frac{x^2}{2 \sigma^2}-\frac{(x - y)^2}{2 h^2}\right\} dxdy.
\end{equation}
To compute this double integral, introduce the generating function for
the standard Hermite polynomials: 
\[
  g\sub{st}(z, s) = e^{2 z s - s^2} = \sum_{n = 0}^{\infty} \frac{s^n}{n!} H_n(z),
\qquad\text{so that}\qquad
  H_k(z) = \lim_{s \rightarrow 0} \frac{\partial^k}{\partial s^k} g\sub{st}(z, s).
\]
In a similar manner, introduce the function $g(z, s, \sigma) =
\exp\{(2 z s - s^2)/(2 \sigma^2)\}$, whence
\begin{equation}\label{eq:Hermite-polys}
  Hn_k(z, \sigma) = \frac{\sigma^k}{\sqrt{k!}} \lim_{s \rightarrow 0} \frac{\partial^k}{\partial s^k} g(z, s, \sigma).
\end{equation}
The resulting integral over $\R^2$ can be computed explicitly, so that exchanging the order of limits, differentiation, and integration,
implies that \eqref{eq:gausssens} results in the simplified form:
\begin{align*}
  s_{jk} &= \delta_{j k} \frac{\sigma^k}{\sqrt{k!}} \frac{(\tau^2)^{j/2}}{\sqrt{j!}} \,k!\, (\tau^2)^{-k} = \delta_{j k} \frac{\sigma^j}{(\tau^2)^{j/2}} = \left(1 + \frac{h^2}{\sigma^2}\right)^{-j/2}\delta_{jk} .
\end{align*}
That is, the sensitivity matrix in this case is diagonal and,
therefore so is $S S^T$. According to \eqref{eq:edof}, the corresponding EDoF for this all-Gaussian
case, $\nu_G(h)$, is then simply:
\begin{equation}\label{eq:edofg}
  \nu_G(h) = \sum_{j=1}^{\infty}s_{jj}^2 = \sum_{j=1}^{\infty} \left(1 + \frac{h^2 }{\sigma^2 }\right)^{-j} = \frac{\sigma^2}{h^2}.
\end{equation}
The calculation of the EDoF of~\cite{ref:McCloud2020} is very simple in this case,
and leads to the relationship already seen in the previous section that
$\nu_3(h)=1+\nu_G(h)$. 

\begin{remark}
In our experience it is very unusual for the sensitivity matrix
$\{s_{jk}\}$ to be diagonal, as is the case in this all-Gaussian
example. Everything else we tried led to more complicated structures.
\end{remark}

Using, for example, the AMISE normal scale
plug-in rule for the bandwidth of $\hh = \left(4/3\right)^{1/5} \sigma
n^{-1/5}$, which assumes $f$ to be (correctly in this case) normal
with variance $\sig^2$, implies $\nu_G(\hh)= \left(3/4\right)^{2/5}
n^{2/5}$. This begs the tantalizing question of whether there is a
plug-in rule that minimizes the Kullback–Leibler divergence instead
  of the AMISE. Invoking the AMKLD criterion from Section
  \ref{sec:amkld}, we can in fact answer that question. The following
  theorem provides the necessary basis for that.
  
\begin{theom}\label{the:Gaussian-amkld}
Consider the convolution kernel for the all-Gaussian case where $f(x)=\phi(x;\sig)$
and $f_K(x)=\phi(x;\tau)$ discussed at the beginning of the section
with $\tau^2:=\sig^2+h^2$, and let $r=h/\sig$. Then, the three terms
in \eqref{eq:amkld-def} that define AMKLD are as follows: 
\[
B(r) = \frac{1}{2} \left( \log\left( 1 + r^2\right) - \frac{r^2}{1 +
    r^2}\right), \qquad V(r)= \frac{1}{2 nr^2},
\]
and
\[
BV(r)= \frac{1}{2 nr^2} \left(\frac{(1 + r^2)^2}{ \sqrt{(2 + r^2)^2 - 1}} - 1\right)-\frac{1}{2 n},
\]
so that the all-Gaussian AMKLD is: $\text{AMKLD}_G(r)=B(r)+V(r)+BV(r)$.
\end{theom}
\begin{proof}
See appendix.
\end{proof}

Differentiating $\text{AMKLD}_G(r)$ and equating to zero, implies that
the optimal value of $r$, $r_\ast=h_\ast/\sig$, satisfies the equation
\begin{equation}\label{eq:optgamkl-r}
   n r_\ast^6 (3 + r_\ast^2 ) \sqrt{3 + 4 r_\ast^2  + r_\ast^4}  = 3 (1 + r_\ast^2)^3.
\end{equation}
For large $n$ and small $r$ we can approximate each side
of \eqref{eq:optgamkl-r} with their corresponding dominant terms,
leading to $3nr_\ast^6\sqrt{3} \approx 3$, from which we deduce the following new normal scale plug-in rule:
\begin{equation}\label{eq:optgamkl-h}
  h_\ast \approx \hh = (1/3)^{1/12}\sig n^{-1/6}.
\end{equation}

Thus AMKLD suggests an optimal bandwidth rule of order $n^{-1/6}$, as
opposed to the AMISE-optimal $n^{-1/5}$. For subtle reasons that we
discuss in Appendix~\ref{app:further-amkld}, this results in an
oversmoothing approach to the optimal bandwidth determination problem,
but the availability of this upper bound is useful. It
would be interesting to speculate whether this result extends in
greater generality beyond the all-Gaussian case.

\section{Numerical Studies}\label{sec-numerical}

In the subsections below, we detail some numerical studies carried out to investigate various characteristics of
the proposed EDoF; these would otherwise be difficult
(if not impossible) to unveil analytically.  In practical implementation, we have found it expedient to utilize the
KDE method based on the diffusion kernels devised by \cite{botev2010kernel}. These
generalized kernels are especially useful in compact supports due to
their automatic boundary correction properties, and they satisfy all
the properties P1 through P5 mentioned in the Introduction.
Computation
of the integral in \eqref{eq:directSampleEdof} that defines the
empirical EDoF is then facilitated  by
the fact that the supports of $f$ and $f_K$ coincide, i.e.,
$a=\ta$ and $b=\tb$. 

The choices for the oracle densities $f(x)$ in these studies
were taken to be the following.
\begin{itemize}
\item[(i)] A 20\%/80\% mixture of a Gaussian with mean 0.4 and standard
deviation 0.07, and an exponential with scale parameter 0.4. The
resulting mixture was then truncated at its $0.99$ quantile, resulting
in a density supported approximately over $[0, 2]$.
\item[(ii)] A 50\%/50\% mixture of two Gaussians, both with standard
deviation $0.5$, and means $\pm 1.5$. The resulting
symmetric and bimodal density was then truncated between its $0.005$ and $0.995$ quantiles.
\end{itemize}
Diffusion
kernels were used to produce all the KDEs in the ensuing numerical
studies whose densities are compactly supported. For the normal
density study, the Gaussian kernel was employed.

\subsection{Behavior of oracle and empirical EDoF}
In this study we analyzed how the oracle and empirical EDoF measures behave as
functions of the bandwidth. From a sample of $n=2{,}000$ data points, we
considered a range of bandwidths centered around the AMISE-optimal
value, and constructed the KDE based on each bandwidth. For each KDE,
we then obtain the oracle and empirical EDoF measures defined in
\eqref{eq:app-directOracleEdof} and
\eqref{eq:directSampleEdof}. The upper panels of Figure \ref{fig:Fig1-Simulations} show, on a
logarithmic scale, how both  measures decrease
monotonically as the bandwidth increases, thus confirming the
proposed EDoF definition to be a sensible measure of model
complexity. The dotted lines display also the $\hnu_3$
of~\cite{ref:McCloud2020}, and serves as a check to dispel any
suspicions aroused by Sections~\ref{sec-histogram} and \ref{sec-allgauss} that it is always linearly related to $\hnu$.

\subsection{Behavior of kernel sensitivity matrix}
We investigated the diagonal elements of the matrix $\widetilde{S}\widetilde{S}^T$, whose trace comprises the definition of
EDoF in \eqref{eq:edof}, as a function of the matrix dimension. The maximum
values for the OPS degrees were set at 50 and 150, for the
Normal-Exponential and Normal-Normal mixtures, respectively. (Since a
given OPS degree determines the matrix dimension, these maximum  degrees are
also the maximum matrix dimensions.) The maximum degrees are
chosen at the onset of numerical instabilities in the
calculation of the recurrence coefficients for the OPS
that uses the sample empirical density function as the weight.
The lower
panels of Figure \ref{fig:Fig1-Simulations} typify the results, with the
diagonals decaying rapidly to zero.

\begin{figure}[htb!]
\centering
\includegraphics[width=\textwidth]{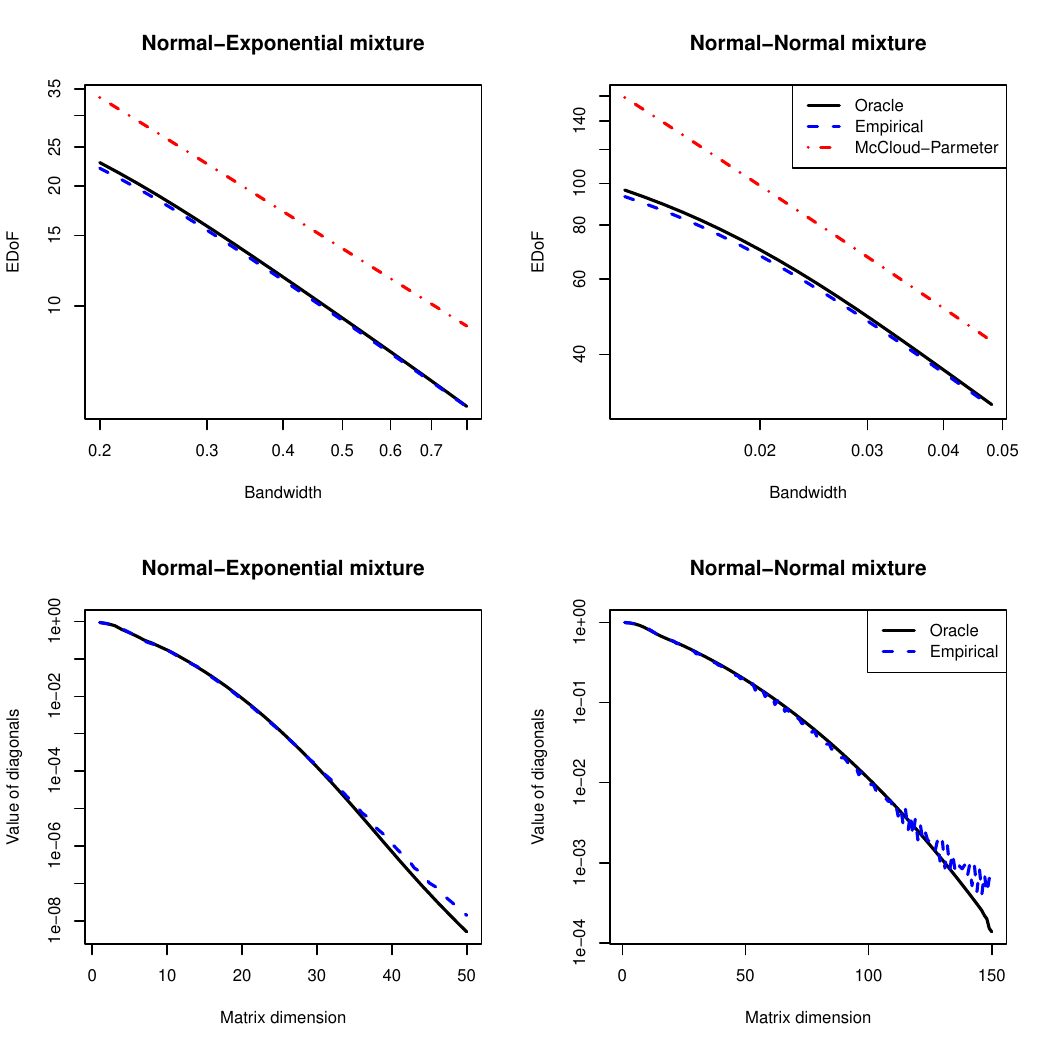}
\caption{Behavior of oracle and empirical EDoF as a function of
  bandwidth (upper panels), and of the diagonals of the 
  matrix $\widetilde{S}\widetilde{S}^T$as a function of matrix dimension (lower
  panels). The bandwidth values in the lower panels were set at $h=0.8$
and $h=0.048$, for the Normal-Exponential and Normal-Normal mixtures,
respectively. All empirical results are based on a sample size of
$n=2,000$.}
\label{fig:Fig1-Simulations}
\end{figure}

\subsection{Bias and variance of EDoF}
Here we generated 1,000 replicates from the distributions over a
representative set of sample sizes ranging from $n=200$ to $n=10{,}000$. For
a given sample size, we then obtain the (empirical) bias and the variance of the
empirical EDoF, $\hnu$. The oracle values of $\nu$ are computed with
the AMISE-optimal value for the bandwidth (based on the oracle density
itself). The results, presented in the top panels of Figure \ref{fig:Fig2-Simulations}, confirm the expected result that both bias and variance gradually
decrease as sample size increases, with larger magnitudes in the
Normal-Exponential case. Interestingly, the bias is negative in both
cases. Based on
a sample of size $n$, it is not possible to produce an OPS with degree larger than $n/2$ because
more than $n$ recurrence coefficients would be needed. 
Thus  the EDoFs associated
with high order polynomials do not exist, and perhaps
this is what is causing the negative bias.

\begin{figure}[htb!]
\centering
\includegraphics[width=\textwidth]{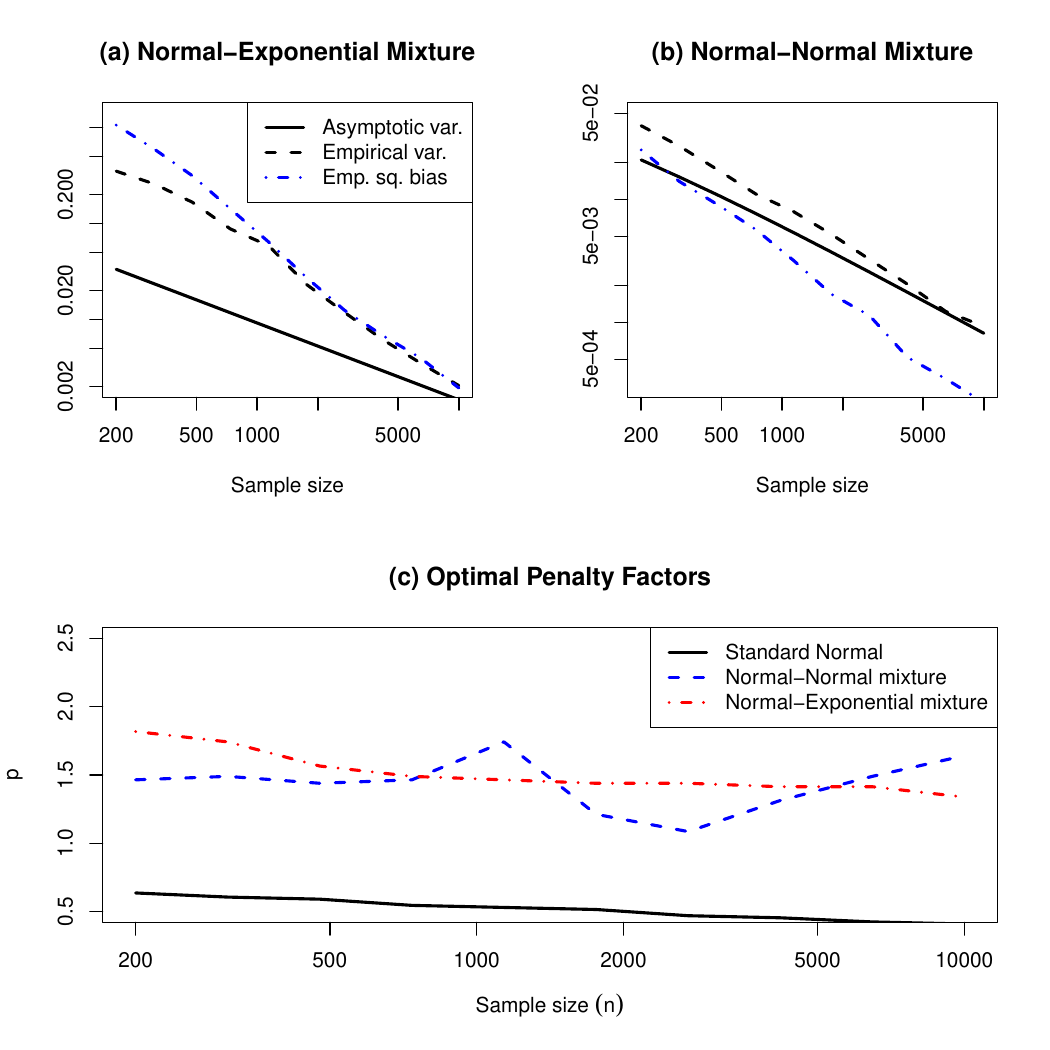}
\caption{Upper panels: squared bias and variance of the empirical EDoF as a function of sample size (logarithmic scale). Lower panel: optimal AIC penalty factor for bandwidth selection as a
  function of sample size (logarithmic scale).}
\label{fig:Fig2-Simulations}
\end{figure}

\subsection{Optimal penalty factors for information criterion-based bandwidth selection}
Explicit calculations in Section~\ref{sec-allgauss}
for the all-Gaussian case led to a new plug-in rule based on an
asymptotic approximation, AMKLD. However,
these calculations, already quite difficult there, are infeasible in
greater generality, and thus one  would hope that minimization of an
information criterion-based quantity such as AIC would provide an
adequate substitute.
Therefore, we now focus  on the development of an
optimized AIC-like criterion for bandwidth selection. That
is, we seek to
investigate the penalty factors $p$ that minimize
\begin{equation}\label{eq:aic}
AIC(h) = p \nu(h) - \ell(h), \qquad \ell(h) = \sum_{i=1}^n\log \hf(x_i|h).
\end{equation}
Since the EDoF
term $\nu(h)$ is an integral part of AIC, that derivation had to come
first, and hence the title of the paper. (Note that $p=1$ in the classical applications of AIC.) 

Through numerical studies, we obtained for each of three distributions (the two
mixtures as well as the standard normal), and for a range of sample
sizes $n$, the bandwidth value that minimizes the KL distance between
the KDE and the true density. Motivated by AMISE optimality
considerations, we then fit the bandwidth as a
polynomial function of $x=n^{-1/5}$. The fitting is via weighted
least-squares with weights inversely proportional to the bandwidth uncertainty (itself assumed to be proportional to the bandwidth). We thus have a function which, for a given
distribution and sample size, returns the KL-optimal value for the KDE
bandwidth $h_\ast$. For example, for the $\mN(0,1)$, we obtained the
cubic $h_\ast=1.2463x-0.5902x^2+0.7349x^3$.

We now differentiate $AIC(h)$ with respect to $h$, evaluate it at 
the KL-optimal bandwidth, and aim to find the penalty factor which
minimizes it. That is, we solved for $p$ in the equation:
\[
  \left.\frac{\partial}{\partial h}AIC(h) \right|_{h = h_\ast} = 0.
\]
Derivatives of $\ell(h)$  and $\nu(h)$ were obtained
numerically. (The Gaussian case is facilitated by the fact that from
\eqref{eq:edofg} one has that $\partial \nu(h)/\partial h=-2\sig^2/h^3$.)
The results, displayed on the bottom panel of Figure \ref{fig:Fig2-Simulations}, do not deviate
far from the value of $p=1$ in classical applications, and suggest little
dependence of the optimal factor on sample size.

\section{A Real Example: the Old Faithful Data}\label{sec-realdata}
To illustrate some real EDoF calculations,  we consider the waiting
time in minutes between successive eruptions of the Old Faithful geyser in
Yellowstone National Park. These data were popularized in the KDE
literature by \citet{azzalini1990geyser} to exemplify a
distinctly bimodal distribution, with peaks at approximately 55 and 80
minutes. There are many similar versions in usage; the one we selected 
consists of $n=272$ observations available as data set \texttt{faithful} in R.

The top left panel of Figure~\ref{fig:Fig3-OldFaithful} displays the data as a rug plot, as well
as three KDE's, obtained using different bandwidth selection rules
(all with the default Gaussian kernel). The values of these and a few other bandwidths
can be seen in Table~\ref{tab:OldFaithful-bandwidths}. The regularized
likelihood cross-validation (CV) method maximizes the product of the leave-one-out
density estimates, and is described in \cite[\S
4.2]{amali2017lorpe} ($\alpha$ is a regularization parameter
that circumvents the problematic situation of zero density
estimates). The AMKLD rule of thumb uses \eqref{eq:optgamkl-h}, with
the usual empirical estimate of the standard deviation.
\begin{figure}[htb!]
\centering
\includegraphics[width=\textwidth]{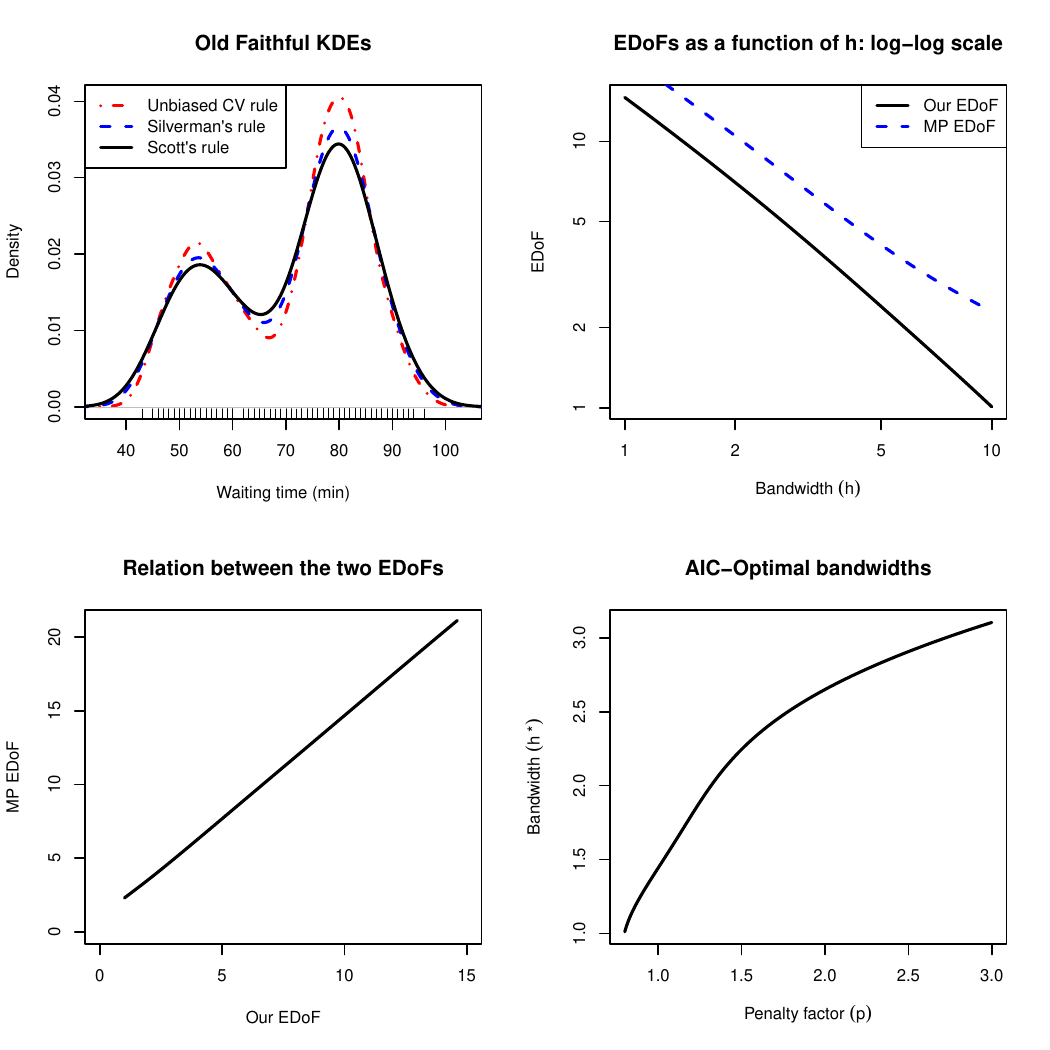}
\caption{Old Faithful waiting times. Panel 1: KDEs produced by
  different bandwidth selection rules. Panel 2: our $\hnu$ (solid line) and McCloud \&
  Parmeter's $\hnu_3$ (dashed line), as a function of
  bandwidth. Panel 3: $\hnu$ vs.~$\hnu_3$ (parametrized by
  $h$). Panel 4: AIC-optimal bandwidth as a function of the AIC penalty factor.}
\label{fig:Fig3-OldFaithful}
\end{figure}

\begin{table}[htb!]
\centering
\caption{Bandwidths for the Old Faithful waiting times produced by
  different methods (Gaussian kernel). The middle column lists the
  \texttt{bw} argument name in the R function \texttt{density}
  that implements the rule given in the first column.} 
\label{tab:OldFaithful-bandwidths}
\centering 
\begin{tabular}{lcccc} 
\toprule
\textbf{Method} & \textbf{R Name} & \textbf{Bandwidth} & $\hnu$ &
                                                                    $\hnu_3$ \\	
\midrule
Regularized likelihood CV ($\alpha=0.5$) & ---  & 2.2550  & 6.1 & 9.3 \\
Biased CV        & \texttt{bcv}  & 2.5976 & 5.2 & 8.0 \\
Unbiased CV      & \texttt{ucv}  & 2.6582 & 5.1 & 7.8  \\
Silverman's rule & \texttt{nrd0} & 3.9876 & 3.2 & 5.1  \\
Scott's rule     & \texttt{nrd}  & 4.6965 & 2.6 & 4.3  \\
AMKLD ($\hat{\sig}=13.59$) & ---  & 4.8737  & 2.5 & 4.2  \\ 
\bottomrule
\end{tabular}
\end{table}

In order to illustrate the values of the empirical EDoF's for a variety of
bandwidths, we display in the top right panel of Figure~\ref{fig:Fig3-OldFaithful} the two
essentially distinct EDoF's: our $\hnu$ and  McCloud \&  Parmeter's
$\hnu_3$. Over the range of suggested bandwidths  in
Table~\ref{tab:OldFaithful-bandwidths},  $\hnu$ varies from approximately 3 to
6, with $\hnu_3$ values between 2 to 3 units larger. The almost linear
relation on the bottom left panel is perhaps unsurprising, given
that the oracle calculations in Section \ref{sec-allgauss} suggest
such a relationship for Gaussian densities and kernels (the data
look Gaussian, albeit a mixture).

The bottom right panel of Figure~\ref{fig:Fig3-OldFaithful} plots the
AIC-optimal bandwidth as a function of the AIC penalty factor. More
specifically, taking $\ell(h)$ to be the empirical log-likelihood as defined
in~\eqref{eq:aic} and $\hnu(h)$ our empirical EDoF, let 
\begin{equation}\label{eq:emp-aic}
AIC(h) = p \hnu(h) - \ell(h).
\end{equation}
For each fixed $p$, we then minimize $AIC(h)$
(with quadratic interpolation on a grid of $h$ values near
the minimum), in order to identify the AIC-optimal bandwidth,
$h_\ast(p)$. The panel displays the pairs $(p,h_\ast(p))$. It is
interesting to note that for $p\approx 1.5$, as for the mixtures in
the bottom panel of Figure \ref{fig:Fig2-Simulations}, the AIC-optimal bandwidth is
about 2.25, which is similar to the bandwidth obtained from
regularized likelihood CV.

\section{Discussion}

We have proposed an oracle definition of EDoF in KDE by
relating the coefficients in OPS expansions of the
empirical density, $\rho(x)$, to the actual KDE, $\hf(x)$, via a
matrix. By viewing the OPS for $\rho(x)$ as collection of uncorrelated
``noise sources'', this \emph{sensitivity} matrix dictates how these noise sources are
suppressed by smoothing, whence taking the trace of a positive semi-definite
and normalized version leads to a sensible definition of EDoF,
as commonly used in the literature on linear smoothing. A proof of
concept is established by demonstrating that, when a histogramming
kernel is employed (so that KDE yields a histogram), EDoF outputs
the correct value: the number of bins minus one. An empirical version
of EDoF follows by the plug-in principle, whence
asymptotic results are then straightforwardly established. 

Some interesting results are obtained when connecting our EDoF to the
only other versions appearing in the literature: two from
\cite{ref:Loader1999} and one from \cite{ref:McCloud2020}. First, we
point out that these three competing versions are empirical, with no
known oracle counterparts. We showed that the influence function based
version of \cite{ref:Loader1999} is approximately the same as that
proposed by \cite{ref:McCloud2020}. This is perhaps not surprising,
since the idea behind both is to formulate the problem in a
regression context, whence the appropriate value is obtained by taking
the trace of the smoothing matrix. Since the EDoF is in the form of
an empirical average, an oracle version is immediate. More surprising
is the fact that the variance function based
version of \cite{ref:Loader1999} is algebraically similar to our
EDoF. 

Outside of the simple histogram, we are able to provide a closed-form
expression for our oracle
EDoF only in the instance when both density and kernel are 
Gaussian. In this case it is
simply the ratio of the variance of the density to the square of the
bandwidth, which then decays sensibly to zero as the latter
grows. An interesting related question was whether minimization of KL divergence would be
feasible as an alternative to more traditional MISE-based bandwidth
selection rules, and we showed that this is indeed the
case with our AMKLD expression which holds for a general density and
kernel. Although this is not sufficiently tractable to be
analytically optimized to yield a plug-in rule, it is in the special
case  when both density and kernel are 
Gaussian, and this yielded a new normal scale rule in which the
optimal bandwidth is $O(n^{-1/6})$. This rule suggests larger 
values than the AMISE optimal $O(n^{-1/5})$ bandwidth, providing a useful upper bound. 

The fact that the EDoF may be a plausible measure of model complexity
leads to the possibility of a new tuning parameter selection rule based
on an information criterion such as AIC, and would therefore  offer a
competing approach to the work of \cite{ref:Loader1999} in
the context of local likelihood for bandwidth selection. (Note however that our development is more
general as it is not restricted to convolution kernels.)  Some numerical work in
inverting this model-selection problem
suggests that the optimal penalty factor (the multiplier of EDoF in
the AIC formulation) is fairly insensitive to the sample
size, but somewhat sensitive to the
nature of the underlying density. Noting that our EDoF captures only
the variance portion of the log-likelihood  (due to its similarity with $\nu_2$), perhaps an
additional piece accounting for bias would improve its performance in
this paradigm.

\section*{Declarations}
This work was partially supported by US Department of Energy grant
DE-SC0015592. The authors have no competing interests to declare that are relevant to the content of this article.


\bibliographystyle{apalike}
\bibliography{edof.bib}

\appendix
\section{Appendix}

\subsection{Proof of Lemma \ref{the:General-Parseval-Identity}}
Since $\{P_k(x)Q_k(y)\}$ forms a complete system, we can appeal to the approximation:
$g(x,y) \approx  \widehat{g}_N(x,y) = \sum_{j,k=0}^Ns_{jk}
P_j(x)Q_k(y)$, with the property that
\[
\lim_{N\rightarrow\infty}\int_{\mY}\int_{\mX} \left[\widehat{g}_N(x,y)-g(x,y)\right]^2dF_P(x)
  dF_Q(y) = 0.
\]
A proof of this in the univariate  case  can be found in, for example,
\cite[Ch.~10]{severini2005elements}. The bivariate (and indeed
multivariate) case follows by a similar argument.  The result hinges crucially on
the completeness of the OPS and the so-called Stone-Weierstrass
Theorem. This theorem is now known to hold in general when applied to any continuous real-valued function
supported on a compact Hausdorff space \citep[Ch.~13]{lax2002functional}, which
includes any closed subset of (finite dimensional) Euclidean space.  
With this setup in place, derivation of the
generalized Parseval Identity is straightforward, by noting that both 
\[
  \int_{\mY} \int_{\mX} \widehat{g}^2(x,y)dF_P(x) dF_Q(y) = \sum_{j,k=0}^{\infty}s_{jk}^2 = \int_{\mY} \int_{\mX} g(x,y)\widehat{g}(x,y)dF_P(x) dF_Q(y),
\]
yield the same result as \eqref{eq:app-General-Parseval-Identity}.

\subsection{Proof of Theorem \ref{the:edof-nu}}
Noting the equivalence of the four versions of $s_{jk}$, make the substitutions $w = F(x)$ and $z = F_K(y)$ in \eqref{eq:salllegendre}, whence
$f(x)dx=dw$ and $f_K(y)dy=dz$, resulting in Legendre
OPS's on $[0,1]$, so that the sensitivity matrix elements become:
\begin{equation}\label{eq:app-bivarCoeffs}
  s_{jk} = \int_0^1 \int_0^1 \frac{K(x(w), y(z))}{f_K(y(z))} L_j(z) L_k(w) dw dz.
\end{equation}
Now apply Lemma~\ref{the:General-Parseval-Identity} with $g(w,
z)=K(x(w),y(z))/f_K(y(z))$ to obtain
\[
\sum_{j, k = 0}^{\infty} s_{jk}^2 = \int_0^1 \int_0^1 g^2(w, z)dw dz,
\quad\Longrightarrow\quad 
  \nu = \sum_{j, k = 1}^{\infty} s_{jk}^2 = \int_0^1 \int_0^1 g^2(w, z) dw dz - 1.
\]
Switching back to the original $x$ and $y$  variables yields the
result in \eqref{eq:app-directOracleEdof}.

\subsection{Proof of Theorem \ref{the:clt-effdf}}
For notational simplicity, we suppress the dependence on $h$. We start
with the IF for $r(y)=\mu_2(y)/\mu_1(y)$, call it $L_4(x;y)$, which is obtained by
applying the chain rule:
\[
  L_4(x;y) = \frac{\partial r(y)}{\partial
    \mu_1(y)}L_1(x;y)+\frac{\partial r(y)}{\partial \mu_2(y)}L_2(x;y) =
  \frac{K^2_h(x,y)}{\mu_1(y)}-\frac{\mu_2(y)}{\mu_1(y)^2}K_h(x,y).
\]
The ultimate goal is the IF for $t_4(F)=\int r(y)dy$. Since the
Gateaux derivative in \eqref{eq:influence-function} can be
interchanged with the integral, we obtain:
 $L_4(x) =  \int_{\ta}^{\tb}L_4(x;y)dy$. 
The Nonparametric Delta Method then gives the stated asymptotic distribution, where:
\[
\omega^2_h = \E_F L_4(X)^2 =   \int_{\ta}^{\tb}\int_{\ta}^{\tb}\E_F
\left\{L_4(X;w)L_4(X;z)\right\}dwdz.
\]
(Note: empirical IFs are obtained by replacing $F\mapsto \hf$,
whence $L(\cdot)\mapsto \ell(\cdot)$.)

\subsection{Proof of Theorem \ref{the:Gaussian-amkld}}
From \eqref{eq:edofg} we obtain immediately the expression for $V(r)$. The bias term can actually be calculated exactly:
\begin{equation}\label{eq:exactbias}
B(r) = -\int_{-\infty}^{\infty} \phi(z, \sigma) \log \left(\frac{\phi(z,\tau)}{\phi(z, \sigma)}\right) dz  = \frac{1}{2} \left(\log\left( 1 + r^2\right) - \frac{r^2}{1 +
    r^2} \right).
\end{equation}
Now, the \emph{extended asymptotic variance}, the sum of
the 2nd and 3rd terms in $\text{AMKLD}_G(r)$, can be expressed as:
\[
  V(r)+BV(r)=\frac{1}{2n}\tr\left(\tS\tS^T(I-\mE)\right)=\frac{1}{2 n}
  \sum_{k=1}^{\infty}\left( 1 + r^2\right)^{-k}(\delta_{kk} - \varepsilon_{kk}).
\]
These terms therefore depend on the difference of matrix elements:
\begin{equation}
  \begin{split}
    \delta_{jk} - \varepsilon_{jk} & = \int_{-\infty}^{\infty} \phi(z, \sigma) Hn_j(z, \tau) Hn_k(z, \tau) dz \\
                                 & = \frac{\tau^{j+k}}{\sqrt{j!k!}} \lim \limits_{\substack{
      s \to 0\\
      t \to 0}} \frac{\partial^j}{\partial s^j} \frac{\partial^k}{\partial t^k} \int_{-\infty}^{\infty} \phi(z, \sigma) g(z, s, \tau) g(z, t, \tau) dz \\
    & = \frac{\tau^{j+k}}{\sqrt{j!k!}} \lim \limits_{\substack{
      s \to 0\\
      t \to 0}} \frac{\partial^j}{\partial s^j} \frac{\partial^k}{\partial t^k} \exp \left\{\frac{2 s t \sigma^2 - h^2 (s^2 + t^2)}{2\tau^4}\right\},
  \end{split}
\end{equation}
so that setting $j=k$,
\[
  \delta_{kk} - \varepsilon_{kk} = \frac{\tau^{2k}}{k!} \lim \limits_{\substack{
      s \to 0\\
      t \to 0}} \frac{\partial^k}{\partial s^k} \frac{\partial^k}{\partial t^k} \exp \left\{\frac{2 s t \sigma^2 - h^2 (s^2 + t^2)}{2\tau^4}\right\}.
\]
Taylor expanding the exponential, letting $C_k^j = k!/(j! (k - j)!)$
denote the binomial coefficients, and noting that only the terms
proportional to $s^k t^k$ contribute, we obtain
\begin{equation}
  \begin{split}
  \delta_{kk} - \varepsilon_{kk} & = \frac{1}{(k!)^2 2^k\tau^{2k}} \lim \limits_{\substack{
      s \to 0\\ t \to 0}} \frac{\partial^k}{\partial s^k} \frac{\partial^k}{\partial t^k} \left(2 s t \sigma^2 - h^2 (s^2 + t^2)\right)^k \\
  & =  \frac{1}{(k!)^2 2^k\tau^{2k}} \lim \limits_{\substack{
      s \to 0\\
      t \to 0}} \frac{\partial^k}{\partial s^k} \frac{\partial^k}{\partial t^k} \sum_{j=0}^k C_k^j (2 s t \sigma^2)^{k - j} (-1)^j h^{2 j} (s^2 + t^2)^j \\
  & = \frac{1}{(k!)^2 2^k\tau^{2k}} \lim \limits_{\substack{
      s \to 0\\
      t \to 0}} \frac{\partial^k}{\partial s^k} \frac{\partial^k}{\partial t^k} \sum_{\substack{j = 0\\ j \mbox{\scriptsize \,even}}}^k C_k^j C_j^{j/2} (2 \sigma^2)^{k - j} (-1)^j h^{2 j} s^k t^k \\
  & = \frac{k!}{(1 + r^2)^{k}} \,\sum_{j=0}^{\floor*{k/2}} \frac{r^{4j}}{2^{2j} (k - 2j)! (j!)^2},
  \end{split}
\end{equation}
which leads to the expression:
\[
  V(r)+BV(r) = \frac{1}{2 n} \sum_{k=1}^{\infty}\frac{k!}{(1 +
    r^2)^{2k}} \,\sum_{j=0}^{\floor*{k/2}} \frac{r^{4j}}{2^{2j} (k -
    2j)! (j!)^2} := \frac{1}{2 n}\sum_{j=0}^{\infty}T_j(r).
\]
Progress can now be made by examining each $T_j(r)$ term, which
involves summing over $k$. The $j=0$ term yields
 a straightforward geometric series, $T_0(r)=[r^2(2 + r^2)]^{-1}$, and the
 remaining terms reduce to a geometric series upon differentiation an
 appropriate number of times. It can therefore  be shown that for $j>0$:
\begin{equation*}
 T_j(r)= \frac{r^{4 j}}{2^{2 j} (j!)^2}\; \lim_{x \rightarrow 1}
 \frac{\partial^{2 j}}{\partial x^{2 j}}\left\{ \frac{x^{2 j} \left(1 + r^2
     \right)^{-4 j}}{1 - x \left(1 + r^2 \right)^{-2}}\right\}
 = \frac{(1 + r^2)^2}{r^2 (2 + r^2)^{2 j + 1}}\,a_j ,\qquad a_j=\frac{(2 j)!}{2^{2 j} (j!)^2},
\end{equation*}
whence we obtain:
\[
  V(r)+BV(r) = \frac{1}{2 n}\left( \frac{(1 + r^2)^2}{r^2 (2 + r^2)} \sum_{j=0}^{\infty} \frac{a_j}{(2 + r^2)^{2 j}} - 1\right)
  =  \frac{1}{2 n} \left(\frac{(1 + r^2)^2}{r^2 \sqrt{(2 + r^2)^2 - 1}} - 1\right).
\]

\subsection{Further Considerations for Gaussian AMKLD}\label{app:further-amkld}
It should be noted that in practice the AMKLD normal scale plug-in
rule arrived at by solving \eqref{eq:optgamkl-r}  will, in the vast majority of cases, result in a bandwidth
that leads to oversmoothing. The reason for this can be understood
as follows. The expression for $\Cov(c_k, c_j)$ given in
\eqref{eq:cj} hinges on the orthogonality property of Hermite polynomials:
\begin{equation} \label{eq:hermiteortho}
   \int_{-\infty}^{\infty} Hn_k(x, \sigma) Hn_j(x, \sigma) \phi(x, \sigma) dx = \delta_{kj}.
\end{equation}
For finite samples, $\phi(x, \sigma)$ is replaced by the empirical
density, and the integral  becomes
\begin{equation} \label{eq:sampleortho}
\Theta_{jk} = \frac{1}{n} \sum_{i=1}^n Hn_j(x_i, \sigma) Hn_k(x_i, \sigma).
\end{equation}
While $\E\Theta_{jk}=\delta_{jk}$ as dictated by \eqref{eq:hermiteortho},
for large polynomial degrees and for practical sample sizes, the
 expectation of $\Theta_{jk}$ will be a poor approximation to its
typical values. Consider, for example, the case $j=50=k$. The
integrand of \eqref{eq:hermiteortho} is shown in the top left panel of
Figure \eqref{fig:Fig4-SamplingDistsSk} (for $\sigma = 1$). This integrand is substantial approximately over the interval $\pm
2\sigma\sqrt{j + 1/2}$, and receives its largest contributions from
the regions near the interval boundaries around $\pm 14 \sigma$. Since
the probability of
getting standard normal random variates as large as $\pm 14 \sigma$ for realistic sample
sizes is extremely small, $\Theta_{jj}$ for $j = 50$ will tend to be
substantially smaller than its expected value of $1$.

\begin{figure}[htb!]
\centering
\includegraphics[width=\textwidth]{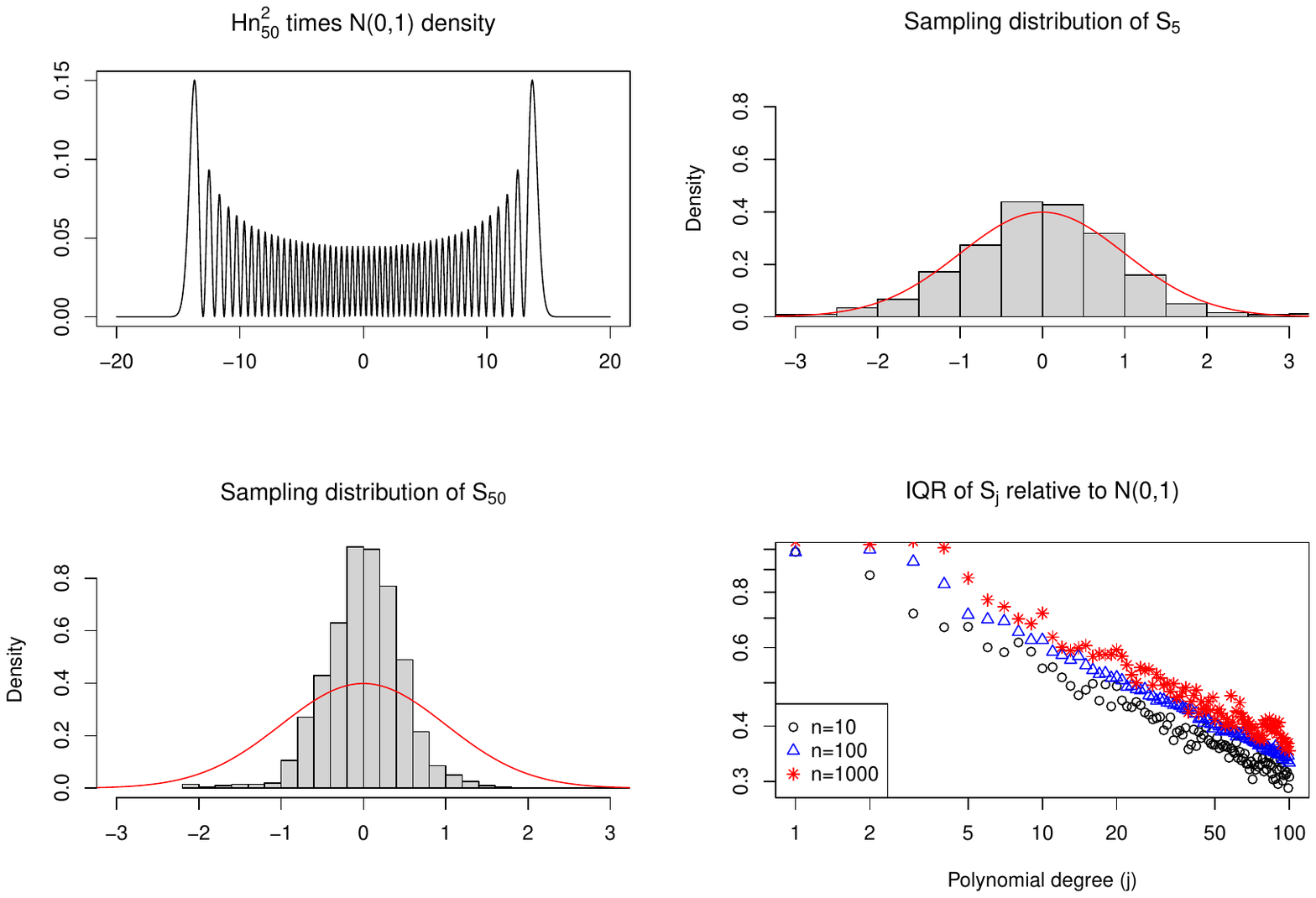}
\caption{Illustration of issues arising in the case of
  Gaussian AMKLD for finite sample sizes. Top left panel: the function in the
  integrand of \eqref{eq:hermiteortho}. Top right and bottom left panels: histograms of sampling
  distributions of $S_{5}$ and $S_{50}$, respectively, along with the
  standard normal density curve. Bottom right panel: interquartile ranges of
$S_j$ relative to the standard normal, as a function of $j$ (both axes
on log scale), for sample sizes of $n=\{10,100,1000\}$.}
\label{fig:Fig4-SamplingDistsSk}
\end{figure}

Consider now the optimal bandwidth rule \eqref{eq:optgamkl-h}. This
rule results in oversmoothing because the expected extended
asymptotic variance term, $V(r)+BV(r)$, derived
in the transition from \eqref{eq:amkld-variance} to \eqref{eq:amkld-def},
overestimates the contribution of variance
into KL divergence  for ``typical'' realistic samples. This
leads to a wider bandwidth that, in the vast
majority of cases, will penalize the $V(r)+BV(r)$ term more 
than what is necessary. For distributions supported on a finite
interval and bounded away from 0, \eqref{eq:amkld-def} is expected to work
well. However,  similar problems are likely to surface for
distributions with longer tails than that of the Gaussian
(for example, the effect is very clear for the exponential and its associated
OPS comprised of Laguerre polynomials). AMKLD for
such cases will be dominated by fluctuations in the tails.

Another way to view this effect is to consider the distributions of
coefficients $c_j$.  While it is true that for $j\geq 1$ we have
$\V(c_j)=1/n$, for large degrees $j$ these distributions become highly leptokurtic, with most of the variance contributed by the long
tails. This effect can be seen in the top right and bottom left panels of Figure
\eqref{fig:Fig4-SamplingDistsSk} which show the sampling distributions
of $S_{j}:=c_j\sqrt{n}$ for $j = 5$ and $j=50$, respectively (both
with $n=1000$). (Note that with this scaling $S_{j}$ has exactly zero mean and
unit variance.) The dramatic increase in kurtosis in going from the
lower to the higher degree is evident. The bottom right panel of Figure
\eqref{fig:Fig4-SamplingDistsSk} illustrates another way to visualize
the effect by plotting $j$ vs.~the interquartile ranges of
$S_j$, $\text{IQR}_j$, 
relative to that for a standard normal (the appropriate limiting distribution
as $n\rightarrow\infty$). The panel therefore displays, on
logarithmic scales,  the pairs of
points $(j,w_j)$, where, letting $\Phi(\cdot)$ denote the cdf of a standard normal,
\[
w_j = \frac{\text{IQR}_j}{\Phi^{-1}(0.75)-\Phi^{-1}(0.25)}.
\]
There are three sets of points on the plot, corresponding to different sample sizes. 
The decline of $w_j$ values with increasing $j$ is another
manifestation that an increasing portion of the mass is going into the tails.
We see that for larger sample sizes the
divergence of $w_j$ away from its nominal asymptotic value of $1$ occurs at
correspondingly larger degrees $j$, hinting at a complex interplay in
the dual asymptotic regimes of $j$ and $n$. The plot suggests the
approximate relationship $w_j \propto j^{-1/4}$ for large $j$.

\end{document}